\documentclass[conference]{IEEEtran}
\IEEEoverridecommandlockouts

\usepackage{cite}
\usepackage{amsmath,amssymb,amsfonts}
\usepackage{algorithmicx}
\usepackage[linesnumbered,ruled,vlined]{algorithm2e}
\usepackage{graphicx}
\usepackage{textcomp}
\usepackage{xcolor}
\usepackage{enumerate}
\usepackage{booktabs}
\usepackage{diagbox}
\usepackage{supertabular}
\usepackage{caption}
\usepackage{float} 
\usepackage{subcaption}
\usepackage{verbatim}
\usepackage{bm}
\usepackage{breqn}
\usepackage{stfloats}
\usepackage{amsthm}
\usepackage{hyperref}

\newtheorem{definition}{Definition}

\newtheorem{lemma}{Lemma}

\def\BibTeX{{\rm B\kern-.05em{\sc i\kern-.025em b}\kern-.08em
    T\kern-.1667em\lower.7ex\hbox{E}\kern-.125emX}}

\begin{document}

\title{ Task Freshness-aware Incentive Mechanism for Vehicle Twin Migration in Vehicular Metaverses
}
\author{\IEEEauthorblockN{Jinbo Wen\IEEEauthorrefmark{1}, Jiawen Kang\IEEEauthorrefmark{1}, Zehui Xiong\IEEEauthorrefmark{2}, Yang Zhang\IEEEauthorrefmark{3}, Hongyang Du\IEEEauthorrefmark{4}, Yutao Jiao\IEEEauthorrefmark{5}, Dusit Niyato\IEEEauthorrefmark{4}, \textit{Fellow, IEEE}}

\IEEEauthorblockA{\IEEEauthorrefmark{1}\textit{Guangdong University of Technology, China}
\IEEEauthorrefmark{2}\textit{Singapore University of Technology and Design, Singapore}\\
\IEEEauthorrefmark{3}\textit{Nanjing University of Aeronautics and Astronautics, China}
\IEEEauthorrefmark{4}\textit{Nanyang Technological University, Singapore}\\
\IEEEauthorrefmark{5}\textit{Army Engineering University of PLA, China}
}

\IEEEcompsocitemizethanks{The work was supported by NSFC under grant No. 62102099, No. U22A2054, No. 62071343, and No. 62101594, and the Pearl River Talent Recruitment Program under grant 2021QN02S643, and supported in part by the National Research Foundation (NRF), Singapore, and Infocom Media Development Authority under the Future Communications Research Development Programme (FCP), and also supported by the Ministry of Education, Singapore, under its SUTD Kickstarter Initiative (SKI 20210204). (Corresponding author: Jiawen Kang (e-mail: kavinkang@gdut.edu.cn)).}
}



\maketitle

\begin{abstract}
Vehicular metaverse, which is treated as the future continuum between automotive industry and metaverse, is envisioned as a blended immersive domain as the digital twins of intelligent transportation systems. Vehicles access the vehicular metaverses by their own Vehicle Twins (VTs) (e.g., avatars) that resource-limited vehicles offload the tasks of building VTs to their nearby RoadSide Units (RSUs). However, due to the limited coverage of RSUs and the mobility of vehicles, VTs have to be migrated from one RSU to other RSUs to ensure uninterrupted metaverse services for users within vehicles. This process requires the next RSUs to contribute sufficient bandwidth resources for VT migrations under asymmetric information. To this end, in this paper, we design an efficient incentive mechanism framework for VT migrations. We first propose a novel metric named Age of Migration Task (AoMT) to quantify the task freshness of the VT migration. AoMT measures the time elapsed from the first collected sensing data of the freshest avatar migration task to the last successfully processed data at the next RSU. To incentivize the contribution of bandwidth resources among the next RSUs, we propose an AoMT-based contract model, where the optimal contract is derived to maximize the expected utility of the RSU that provides metaverse services. Numerical results demonstrate the efficiency of the proposed incentive mechanism for VT migrations.

\end{abstract}


\section{Introduction}
With the gradual maturation of metaverse technologies, implementing metaverse-like immersive experiences within
vehicles appears to be a potential future direction for vehicular interactions \cite{zhou2022vetaverse}. Vehicular metaverse is expected to lead an evolution of the automotive industry\cite{xu2022epvisa}, which integrates extended reality technologies and real-time vehicular data seamlessly to blend physical and virtual spaces for drivers and passengers within vehicles\cite{jiang2022reliable}. In \cite{yu2022bi}, smart driving of the digital twin in the metaverse was introduced. As the digital component of vehicular metaverses, Vehicle Twins (VTs) are large-scale and highly accurate digital replicas that cover the life cycle of vehicles and manage vehicular applications \cite{yu2022bi}. With the help of intra-twin communications, which refer to interactions between VTs and vehicles\cite{10090432}, vehicles can access vehicular metaverses through VTs, for example, in an avatar manner. The VTs can be updated in virtual spaces continuously by sensing data from surrounding environments\cite{zhang2022toward}, including bio-data of passengers, real-time vehicle status, and traffic data in the physical space\cite{xu2022epvisa}, which is advantageous in the development of vehicular metaverses that can interact and coexist with the physical space, functioning as autonomous and durable virtual spaces\cite{yu2022bi}.

Due to the resource limitation of vehicles, it is impractical for vehicles to build high-fidelity virtual models, which may lead to intensive computation for resource-limited vehicles \cite{jiang2022reliable}. Under such conditions, vehicles offload the large-scale rendering tasks of building VTs to the nearby edge servers in RoadSide Units (RSUs) for ultra-reliable and lower-latency metaverse services. Here the RSU providing metaverse services is called Metaverse Service Provider (MSP). Owing to the limited coverage of RSUs, VTs with a mobile nature have to be migrated from the current RSU (i.e., the MSP) to others for continuous metaverse services. Hence, the task freshness of the VT migration (i.e., the time elapsed of completing the current VT migration task) is essential to the provision of continuous metaverse services. To ensure the task freshness of the VT migration, VT migrations require enough available resources, especially bandwidth resources, thus the destination RSUs are required to provide bandwidth resources for VT migrations, where the destination RSUs are called Metaverse Resource Providers (MRPs). Because of information asymmetry, MRPs' private information (e.g., channel conditions and bandwidth costs) might be not aware to the MSP\cite{bolton2004contract}. As a result, a malicious MRP may not contribute bandwidth resources honestly to obtain more benefits without a reasonable incentive mechanism\cite{hou2017incentive}, which affects the task freshness of the VT migration. 

Some efforts have been conducted for optimizing resource allocation and efficiently processing computing-intensive tasks of real-time rendering in vehicular metaverses\cite{jiang2022reliable,xu2022wireless,9827604,xu2023generative}. For example, the authors  in \cite{jiang2022reliable} proposed a hierarchical game-theoretic approach to
investigate the sustainable and reliable coded distributed computing scheme, which supports immersive user experiences in vehicular metaverses. In \cite{xu2022wireless}, the authors formulated a learning-based incentive mechanism to evaluate and enhance VR experiences in the metaverse. In \cite{9827604}, the authors proposed a quantum collective learning and many-to-many matching game-based scheme in the metaverse for connected and autonomous vehicles. However, the above works ignore the VT migration problem due to the mobility of vehicles. Thus, it is still challenging how to optimize resource allocation for VT migrations in vehicular metaverses \cite{yu2022bi}.

To address the above challenges, in this paper, since the existing metrics like Age of Task \cite{song2019age} cannot measure the VT migration delay, we first propose a novel metric named Age of Migration Task (AoMT) based on the concept of Age of Information (AoI). To improve VT migration efficiency, we formulate an AoMT-based incentive mechanism with asymmetric information. The main contributions of this paper are summarized as follows:
\begin{itemize}
    \item To measure precisely the task freshness of the VT migration, we propose a novel metric named AoMT for vehicular metaverses, which can be applied to evaluate the satisfaction of the MSP.
    \item To incentivize the contribution of bandwidth resources among MRPs, we propose an AoMT-based contract model. \textit{To the best of our knowledge, this is the first work studying the incentive mechanism for VT migrations in vehicular metaverses.}
    \item We design the optimal contract which is feasible and maximize the utility of the MSP under information asymmetry. Numerical results demonstrate that the proposed incentive mechanism is practical and efficient.
\end{itemize}

\begin{figure*}[t]\centering     \includegraphics[width=0.9\textwidth]{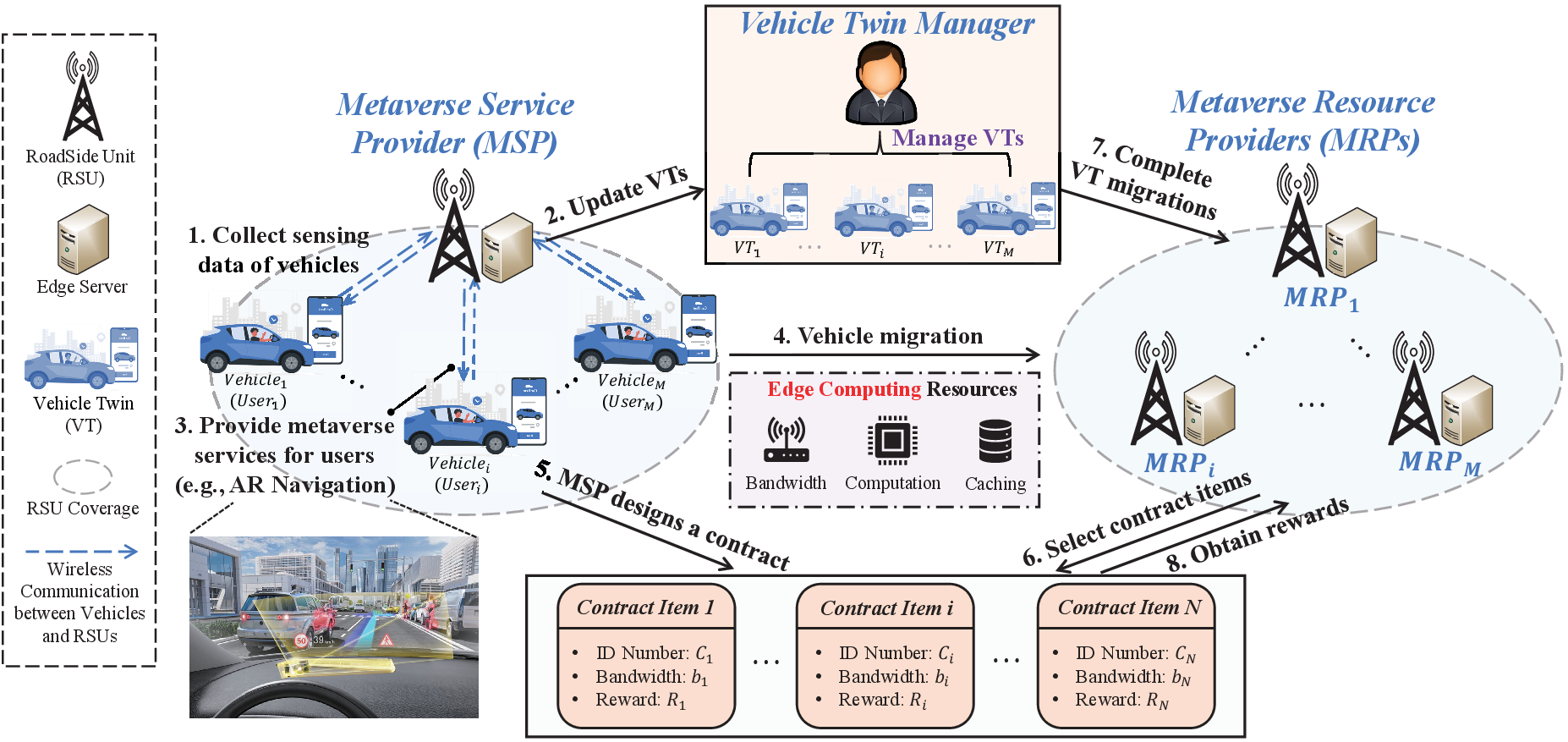}     \caption{An AoMT-based incentive mechanism framework for VT migrations.}  \label{system}     
\end{figure*}

\section{AoMT-based Incentive Mechanism Framework for Vehicle Twin Migration}
As shown in Fig. \ref{system}, edge-assisted remote rendering is an important technology applied in the metaverse\cite{huang2022joint}. To build VTs (e.g., avatars) for accessing metaverse services like Augmented Reality (AR) navigation, occupants (i.e., users) send service requirements to the nearby RSU (i.e., the MSP) that can provide necessary resources (i.e., storage, caching, and computing) for the VT construction\cite{yu2022bi}. For the convenience of explanation, we take vehicle avatars as an example of VTs. Then, the MSP offloads computation-intensive rendering tasks to its proximal edge server and builds avatars to provide lower-latency and ultra-reliable metaverse services for users\cite{huang2022joint}. To efficiently manage avatars on RSUs, Vehicle Twin Managers of the RSU are introduced. However, when the users travel on the road, the current MSP cannot provide continuous services for users outside its coverage. Thus, the avatars have to be migrated to other RSUs. In addition, to ensure immersive experiences for users, the MSP requires sufficient bandwidth resources to enhance avatar migration efficiency and meets the delay requirement of metaverse services during migration\cite{huang2022joint}. We provide more details of the framework as follows:
\begin{itemize}
    \item \textbf{MSP:} The MSP collects sensing data of users and builds avatars to provide ultra-reliable and real-time metaverse services for users. To ensure high-quality metaverse services, the MSP focuses on the task freshness of the avatar migration and requires bandwidth resources from the destination RSUs (i.e., MRPs). After completing avatar migrations, the MSP pays for the MRPs according to their contributions.
    \item \textbf{MRPs:} Each MRP contributes bandwidth resources for the MSP to achieve the avatar migration. The required amount of bandwidth depends on the service level agreements. All MRPs with private information (e.g., channel conditions and bandwidth costs) are selfish and have the potential to obtain more benefits because of information asymmetry. Note that each MRP becomes the new MSP after completing the current avatar migration.
    \item \textbf{Users:} Occupants request and obtain metaverse services from the MSP, such as AR navigation and VR vehicular videos. After completing the avatar migration, each vehicle establishes a connection with the MRP where its avatar is hosted to provide metaverse services for users, and the MRP becomes the new MSP. Note that we consider that each vehicle has a corresponding avatar to manage vehicular applications during migration.
    \item \textbf{Vehicle Twin Manager:} The main responsibility of the Vehicle Twin Manager is to manage avatars on its RSU (i.e., the MSP), including updating avatars. For instance, when avatars experience technical issues, such as being unable to maintain stability, the Vehicle Twin Manager immediately informs the MSP to reconstruct avatars, which ensures the high quality of immersive experiences for users.
\end{itemize}

\section{Problem Formulation}
To incentivize MRPs for the contribution of bandwidth resources, we first propose a novel metric named AoMT to quantify the task freshness of avatar migration, which can be applied to evaluate the satisfaction of the MSP. Second, we formulate the utility functions of both MRPs and the MSP (i.e., the avatar migration task publisher). We consider that there are one MSP and a set $\mathcal{M}$ of $M$ MRPs in avatar migrations, where $\mathcal{M} = \small\{1,\ldots,m,\ldots,M\small\}$. The MSP, which publishes $M$ avatar migration tasks, motivates $M$ MRPs to contribute bandwidth resources in avatar migrations.

\subsection{Age of Migration Task for Avatar Migrations}

AoI has been commonly used as an effective metric to quantify information freshness at the destination. It is defined as the time elapsed since the generation of the last successfully received message containing updated information about its source system, and its minimization depends on the status update frequency\cite{kosta2017age}. However, it does not consider the data processing procedure\cite{ChenYing}. Recent studies like Age of Task and Age of Processing \cite{li2021age} improve the AoI by taking the data processing time into account, but they only consider the scenarios with single-type sensing data and cannot measure the avatar migration delay. Therefore, to quantify the task freshness of the avatar migration, we propose a new metric named AoMT based on the concept of AoI. Similar to \cite{ChenYing}, AoMT is defined as the time elapsed from the first collected sensing data of the newest avatar migration task to the last successfully processed data at the MRP.

The time of completing an avatar migration comprises three parts: 1) The time of collecting sensing data (e.g., traffic conditions and vehicle locations) by the MSP (denoted as $t_c$). 2) The time of sending the avatar data from the MSP to the MRP (denoted as $t_s$). 3) The time of processing received data by the MRP (denoted as $t_p$). For simplicity, MRPs have the same ability to communicate with users and process data\cite{huang2022joint}. Therefore, we consider that $t_c$ and $t_p$ are the same for all avatar migrations, respectively. We set $t_c + t_p = T \in \mathbb{R}^+$ as a constant.

It is considered that the Orthogonal Frequency Division Multiplexing Access (OFDMA) technology is applied in the system, which ensures that all communication channels occupied by different MRPs and the MSP are orthogonal\cite{huang2022joint, zhang2019joint}. For MRP $m\in \mathcal{M}$, given the bandwidth $b_m$ allocated to the MSP, the achievable information transmission rate between the MSP and the MRP $m$ is
\begin{equation}
    \begin{aligned}
        \gamma_m = b_m \log_2\bigg(1+\frac{\rho_s h_m^0 d_{s, m}^{-\alpha}}{N_0b_m}\bigg),
    \end{aligned}
\end{equation}
where $\rho_s$, $h_m^0$, $d_{s,m}$, $\alpha$, and $N_0$ represent the transmit power of the MSP, the unit channel power gain, the distance between the MSP and the MRP $m$, the path-loss coefficient, and the noise power density, respectively\cite{9930881,zhang2019joint}. We define the channel power gain between the MSP and the MRP $m$ as $G_{s,m} = h_m^0 d_{s,m}^{-\alpha}$. Therefore, for the MRP $m$, the AoMT of the avatar migration is 
\begin{equation}
    \begin{aligned}
        A_m (b_m)  = \frac{D_m}{\gamma_m} + T,
    \end{aligned}
\end{equation}
where $D_m$ is defined as the avatar data transmitted to the MRP $m$, including the information of the system configuration, historical running data, and real-time avatar states\cite{9491087}. Note that $A_m(b_m)$ is not a convex function with respect to $b_m$.
\subsection{MRP Utility}
The utility of MRP $m$ is the difference between the received monetary reward $R_m$ and its cost $C_m$ of participating in the avatar migration, which is presented as
\begin{equation}
    \begin{aligned}
        U_m = R_m - C_m.
    \end{aligned}
\end{equation}
Since the cost of bandwidth is from the energy consumption of the transmitted information\footnote{Note that the transmit power is the average power of the transmit signal, and the bandwidth reflects the spectrum of significant frequency components allocated for the transmission of the input signal.}, referring to \cite{hou2017incentive, mohsenian2010autonomous}, $C_m$ is defined as 
\begin{equation}
    \begin{aligned}
        C_m = \mathcal{C}_m(b_m/G_{s,m}),
    \end{aligned}
\end{equation}
where $\mathcal{C}_m(\cdot)$ is used to model the bandwidth cost of MRP $m$, given by
\begin{equation}
    \begin{aligned}
        \mathcal{C}_m(x) = a_m x^2,
    \end{aligned}
\end{equation}
where $a_m > 0$ is the bandwidth cost coefficient. Thus, the utility of MRP $m$ becomes
\begin{equation}
    \begin{aligned}
        U_m = R_m - \frac{a_m}{G_{s,m}^2}b_m^2.
    \end{aligned}
\end{equation}

Due to information asymmetry, the MSP is not aware of each MRP's exact bandwidth cost coefficient and channel gain, but it can sort the MRPs into discrete types and use the statistical distributions of the MRPs' types from historical data to optimize the expected utility of the MSP\cite{kang2022blockchain}. Specifically, we divide the MRPs into different types and define the $n$-th type MRP as
\begin{equation}\label{type}
    \begin{aligned}
        \theta_n \triangleq \frac{G_{s,n}^2}{a_n}.
    \end{aligned}
\end{equation}
Since $a_n > 0$ and $G_{s,n} > 0$, we have $\theta_n > 0$. (\ref{type}) indicates that the larger the channel gain $G_{s,n}$ between the MSP and the $n$-th type MRP, or the lower the unit bandwidth cost coefficient $a_n$, the higher the type of the MRP.

Without loss of generality, the MRPs can be classified into a set $\mathcal{N} = \left\{\theta_n: 1 \leq n \leq N\right\}$ of $N$ types. In an ascending order, the MRPs' types are sorted as $\theta_1 \leq \theta_2 \leq \cdots \leq \theta_N$. In this definition, the higher type MRP has a better channel quality or a lower bandwidth cost coefficient. To facilitate explanation, the MRP with type $n$ is called the type-$n$ MRP. Therefore, based on (\ref{type}), the utility of the type-$n$ MRP is rewritten as
\begin{equation}
    \begin{aligned}
        U_n^C(b_n,R_n) = R_n - \frac{b_n^2}{\theta_n}.
    \end{aligned}
\end{equation}

\subsection{MSP Utility}
Since the large AoMT not only leads to a poor immersive experience for users but also degrades the MSP's satisfaction with the avatar migration, the MSP's satisfaction function obtained from the type-$n$ MRP is defined as\cite{kang2022blockchain}
\begin{equation}
    \begin{aligned}
        S_n = \beta \ln(g(b_n)+1),
    \end{aligned}
\end{equation}
where $\beta > 0$ is the unit profit for the satisfaction of the MSP and $g(\cdot)$ is the performance obtained from the type-$n$ MRP, which is defined as
\begin{equation}
    \begin{aligned}
        g(b_n) = K - A_n,
    \end{aligned}
\end{equation}
where $K$ is the maximum tolerant AoMT. In this paper, we consider that $K$ is not less than $A_n$.

Because of information asymmetry, the MSP only knows the number of MRPs and the distribution of each type but does not know each MRP's private type, namely the exact number of MRPs belonging to each type\cite{hou2017incentive}. Thus, considering that the probability of an MRP belonging to a certain type-$n$ is $Q_n$, subject to $\sum_{n\in \mathcal{N}}Q_n = 1$, the utility of the MSP is
\begin{equation}\label{U_s}
    \begin{aligned}
        U_s(\bm{b},\bm{R}) = \sum_{n\in \mathcal{N}}M Q_n (S_n - R_n),
    \end{aligned}
\end{equation}
where $\bm{b}=[b_n]_{1\times N}$ and $\bm{R}=[R_n]_{1\times N}$ denote the bandwidth and reward vectors for all $N$ types of MRPs, respectively.

\section{Optimal Contract Design}
In this section, we formulate the optimal contract, characterize its feasibility conditions, and provide an optimal solution for the formulated contract.

Since the types of MRPs are private information that is not visible to the MSP, a rational MRP may provide false information maliciously and pretend to be an MRP with a better channel condition and/or a smaller bandwidth cost to cheat for more rewards\cite{hou2017incentive}. To improve the performance of avatar migrations under asymmetric information, the MSP uses contract theory to effectively motivate the MRPs to contribute bandwidth resources.

\subsection{Contract Formulation}
A contract consists of a group of bandwidth-reward pairs (i.e., contract items) provided to the MRPs, which are designed by the MSP to maximize the expectation of the MSP's utility. Each MRP selects the best contract item based on its type to maximize its benefit. The contract item can be denoted as $\Phi = \left\{ (b_n, R_n), n\in \mathcal{N}\right\}$, where $b_n$ is the bandwidth provided by the type-$n$ MRP and $R_n$ is the reward paid to the type-$n$ MRP as the incentive for the corresponding contribution.

To ensure that each MRP optimally chooses the contract item designed for its type, the following Individual Rationality (IR) and Incentive Compatibility (IC) constraints should be satisfied\cite{hou2017incentive}.
\begin{definition}
    (Individual Rationality) The contract item that an MRP should ensure a non-negative utility, i.e.,
    \begin{equation}\label{IR}
        \begin{aligned}
            U_n^C(b_n, R_n) = R_n - \frac{b_n^2}{\theta_n} \geq 0,\:\forall n \in \left\{1,\ldots,N\right\}.
        \end{aligned}
    \end{equation}
\end{definition}

\begin{definition}
    (Incentive Compatibility) An MRP of any type $n$ prefers to select the contract item $(b_n, R_n)$ designed for its type rather than any other contract item $(b_j, R_j), \forall j \in \left\{1,\ldots,N\right\}$, and $j\neq n$, i.e., 
    \begin{equation}\label{IC}
        \begin{aligned}
            R_n - \frac{b_n^2}{\theta_n} \geq R_j - \frac{b_j^2}{\theta_n}, \:\forall n, j \in \left\{1,\ldots, N\right\}, n\neq j.
        \end{aligned}
    \end{equation}
\end{definition}
The IR constraints ensure the participation of MRPs and the IC constraints ensure that each MRP chooses the contract item designed for its specific type to obtain the highest benefits. With the IR and IC constraints, the MSP aims to maximize its expected utility. Therefore, the problem of maximizing the expected utility of the MSP is formulated as
\begin{equation}\label{problem1}
    \begin{split}
        \textbf{Problem 1:}\: &\max\limits_{\bm{b},\bm{R}}\:U_s(\bm{b},\bm{R})\\
        &\:\:\text{s.t.}\:\: R_n - \frac{b_n^2}{\theta_n} \geq 0,\:\forall n \in \left\{1,\ldots,N\right\},\\
        &\qquad R_n - \frac{b_n^2}{\theta_n} \geq R_j - \frac{b_j^2}{\theta_n}, \:\forall n, j \in \left\{1,\ldots, N\right\},\\
        &\qquad b_n \geq 0, R_n \geq 0, \theta_n > 0,\: \forall n \in \left\{1,\ldots,N\right\},
    \end{split}
\end{equation}

\subsection{Optimal Contract Solution}
Since there are $N$ IR constraints and $N(N-1)$ IC constraints in \textbf{Problem 1}, it is difficult to directly solve \textbf{Problem 1}. Therefore, we reformulate \textbf{Problem 1} by the following necessary conditions.
\begin{lemma}
    With information asymmetry, a feasible contract must satisfy the following conditions:
    \begin{subequations}
        \begin{align}
            &\:R_1 - \frac{b_1^2}{\theta_1} \geq 0,\label{IR}\\
            &\:R_n - \frac{b_n^2}{\theta_n} \geq R_{n-1}-\frac{b_{n-1}^2}{\theta_n},\:\forall n\in\left\{2,\ldots,N\right\},\label{LDIC}\\
            &\: R_n -\frac{b_n^2}{\theta_n} \geq R_{n+1} - \frac{b_{n+1}}{\theta_n},\: \forall n \in \left\{1,\ldots,N-1\right\},\label{LUIC}\\ 
            &\: R_N \geq R_{N-1} \geq \cdots \geq R_1,\: b_N \geq b_{N-1}\geq \cdots \geq b_1.\label{IC}
        \end{align}
    \end{subequations}
\end{lemma}
   
\begin{proof}
    Please refer to \cite{hou2017incentive}.
\end{proof}
Constraint (\ref{IR}) is related to the IR constraints. Constraints (\ref{LDIC}), (\ref{LUIC}), and (\ref{IC}) are related to the IC constraints. Constraints (\ref{LDIC}) and (\ref{LUIC}) show that the IC constraints can be transformed into the Local Downward Incentive Compatibility (LDIC) and the Local Upward Incentive Compatibility (LUIC) with monotonicity, respectively\cite{hou2017incentive}.

Based on \textbf{Lemma 1}, the optimal rewards for any allocated bandwidth can be obtained by the following \textbf{Lemma 2}.
\begin{lemma}
    For a feasible set of bandwidth $\bm{b}$ satisfying $b_1 \leq \cdots \leq b_n \leq \cdots \leq b_N$, we can obtain the optimal reward as
    \begin{equation}\label{R_n}
         R_n^{\star} = \left\{
        \begin{split}
            & \frac{b_1^2}{\theta_1},\: n = 1,\\
            & R_{n-1} + \frac{b_n^2}{\theta_n}-\frac{b_{n-1}^2}{\theta_n},\:n = 2,\ldots,N.
        \end{split}
        \right.
    \end{equation}
\end{lemma}
\begin{proof}
    Please refer to \cite{hou2017incentive}.
\end{proof}
Based on the iterative method, the optimal reward in (\ref{R_n}) can be rewritten as
\begin{equation}\label{R_n2}
    \begin{aligned}
        R_n^{\star} = \frac{b_1^2}{\theta_1} + \sum\limits_{i=1}^n\Delta_i,\:n\in \mathcal{N}, 
    \end{aligned}
\end{equation}
where $\Delta_1 = 0$ and $\Delta_i = \frac{b_i^2-b_{i-1}^2}{\theta_i}, \forall i \in \left\{2,\ldots,N\right\}$. By substituting the optimal reward (\ref{R_n2}) into the MSP's utility (\ref{U_s}), we can get the MSP's utility with respect to $\bm{b}$. Therefore, \textbf{Problem 1} is reformulated as 
\begin{equation}\label{problem2}
    \begin{split}
        \textbf{Problem 2:}\: &\max\limits_{\bm{b}}\:U_s(\bm{b})\\
        &\:\:\text{s.t.}\:\:b_1 \leq \cdots \leq b_N,
    \end{split}
\end{equation}
where $U_s(\bm{b}) = \sum_{n\in \mathcal{N}}U_{s,n}=\sum_{n\in \mathcal{N}}M(Q_nS_n-e_nb_n^2)$, and $e_n$ is given by
\begin{equation}
    e_n = \left\{
    \begin{aligned}
        &\frac{Q_n}{\theta_n}+\bigg(\frac{1}{\theta_n}-\frac{1}{\theta_{n+1}}\bigg)\sum\limits_{j=n+1}^NQ_j,\: 1 \leq n < N,\\
        &\frac{Q_N}{\theta_N},\: n = N.
    \end{aligned}
    \right.
\end{equation}

\begin{algorithm}[t]
		\caption{Optimal Contract Design}

		\KwIn{Basic channel parameters $\small\{\rho_s, h_m^0, d_{s,m},\alpha, N_0\small\}$ and MRPs' types $\left\{\theta_n,1 \leq n \leq N\right\}$.}
  
		\KwOut{The optimal bandwidth $\bm{b}^{\star}$ and the optimal reward $\bm{R}^{\star}$.}

	    \For{$n = 1,\ldots,N$}
            {
                Initialize the iteration index $z = 0$, the step size $\varphi$, the empty vector $\bm{v}_{s,n}$, and the feasible range of bandwidth $[b_{min}, b_{max}]$, where $b_{min} = b_n^z = 10^5$.

                \While {$b_n^z < b_{max}$}
                {
                    Calculate $U_{s,n}(b_n^z)$.
                    
                    Set $\bm{v}_{s,n}(z) = U_{s,n}(b_n^z)$.

                    $b_n^z = b_n^z + \varphi$.

                    $z = z + 1$.
                }

                Obtain the optimal bandwidth $b_n^{\star}$ for the type-$n$ MRP by using the maximum value index in $\bm{v}_{s,n}$.

            }
            Obtain the optimal bandwidth vector $\bm{b}^{\star '} = \small\{b_1^{\star},\ldots,b_n^{\star},\ldots,b_N^{\star }\small\}$.

            \If {$\bm{b}^{\star '}$ does not satisfy the monotonicity condition}
            {
                Apply \textit{Bunching and Ironing} algorithm \cite{gao2011spectrum} to adjust $\bm{b}^{\star '}$ and output $\bm{b}^{\star}$.

            }
            \Else{
            $\bm{b}^{\star} = \bm{b}^{\star '}$.
            }

            \For{$n = 1,\ldots,N$} 
            {
                Calculate the optimal reward $R_n^{\star}$ based on (\ref{R_n}).
            }

            Obtain the optimal reward vector $\bm{R}^{\star } = \small\{R_1^{\star},\ldots,R_n^{\star},\ldots,R_N^{\star }\small\}$.

		\textbf{return} $\small\{\bm{b}^{\star}, \bm{R}^{\star}\small\}$.
	\end{algorithm}

Since $U_s$ is not a concave function, which cannot be solved by the standard convex optimization tools, we propose a greedy algorithm to design the optimal contract referring to \cite{kang2022blockchain}. Motivated by the above analysis, the detailed contract design
is shown in \textbf{Algorithm 1}. Firstly, we can obtain the optimal bandwidth $b_n^{\star}$ by using the iterative method. If $\bm{b}^{\star '}$ cannot satisfy the monotonicity constraint, the iterative algorithm, i.e., \textit{Bunching and Ironing} algorithm\cite{gao2011spectrum} is adopted to obtain the optimal solution $\bm{b}^{\star}$, which ensures that the monotonicity constraint is satisfied. Finally, the optimal reward $R_n^{\star}$ can be calculated by (\ref{R_n}). Note that the computational complexity of \textbf{Algorithm 1} is $\mathcal{O}(N\log\big(\frac{b_{max}-b_{min}}{\varphi}\big))$, which indicates that \textbf{Algorithm 1} is actually efficient.

\begin{table}[t]\label{parameter}
	\renewcommand{\arraystretch}{1.1} 
	\caption{ Key Parameters in the Simulation.}\label{table} \centering 
	\begin{tabular}{m{4.7cm}<{\raggedright}|m{2.7cm}<{\centering}} 
		\hline		
		\textbf{Parameters} & \textbf{Values}\\	
		\hline
		Transmit power of the MSP $(\rho_s)$ & $23$\:$\rm{dBm}$  \\	
		\hline
		Noise power density $(N_0)$ & $-174$\:$\rm{dBm/Hz}$  \\
		\hline
		Path-loss coefficient $(\alpha)$  &  $2$  \\	
		\hline		
		The size of avatar data transmitted to the MRP $m$ $(D_m)$ &  $[100\:\rm{MB},200\:\rm{MB}]$ \\
		\hline		
		Unit profit for the satisfaction $(\beta)$ &  $200$\\
		\hline
		Unit bandwidth cost $(a_m)$ & $[0.0001,0.001]$\\
		\hline
		Maximum tolerant AoMT $(K)$ & $50$\:$\rm{s}$\\
		\hline
		Distance between the MSP and the MRP $m$ $(d_{s,m})$ & $500\:\rm{m}$\\
		\hline
		The sum of the time of collecting data and processing data $(T)$ & $5$\:$\rm{s}$\\
		\hline
	\end{tabular}\label{parameter}
\end{table}

\section{Numerical Results}
In this section, we consider $M = 10$ MRPs and the type-$n$ follows the uniform distribution\cite{kang2022blockchain}. Referring to \cite{wang2019interference,maccartney2014omnidirectional, kang2022blockchain,hou2017incentive,8031051}, the main parameters are listed in Table \ref{parameter}. Firstly, we validate the IC and IR constraints. Then, we compare the proposed incentive mechanism with other incentive mechanisms: 
\begin{enumerate}[1)]
    \item \emph{\textbf{Contract theory with complete information}} that the channel condition and the bandwidth cost of MRPs are known by the MSP \cite{hou2017incentive}.
    \item \emph{\textbf{Contract theory with social maximization}} \cite{xiong2020multi} that the MSP aims to maximize social welfare under information asymmetry \cite{ye2022incentivizing}.
\end{enumerate}

\begin{figure}
	\centering
	\includegraphics[width=0.48\textwidth]{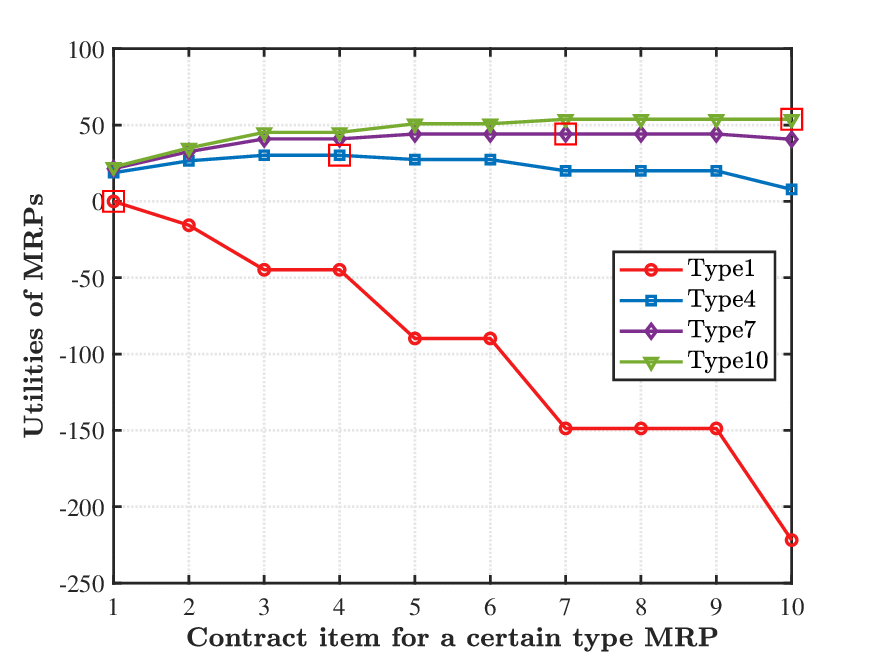}
	\caption{Utilities of MRPs under different types.}\label{IC_IR}  
\end{figure} 

\begin{figure*}[t]
    \begin{center}
	\begin{minipage}[t]{0.48\linewidth}
		\centering
		\captionsetup{font={normal}}
		\includegraphics[width=1\linewidth]{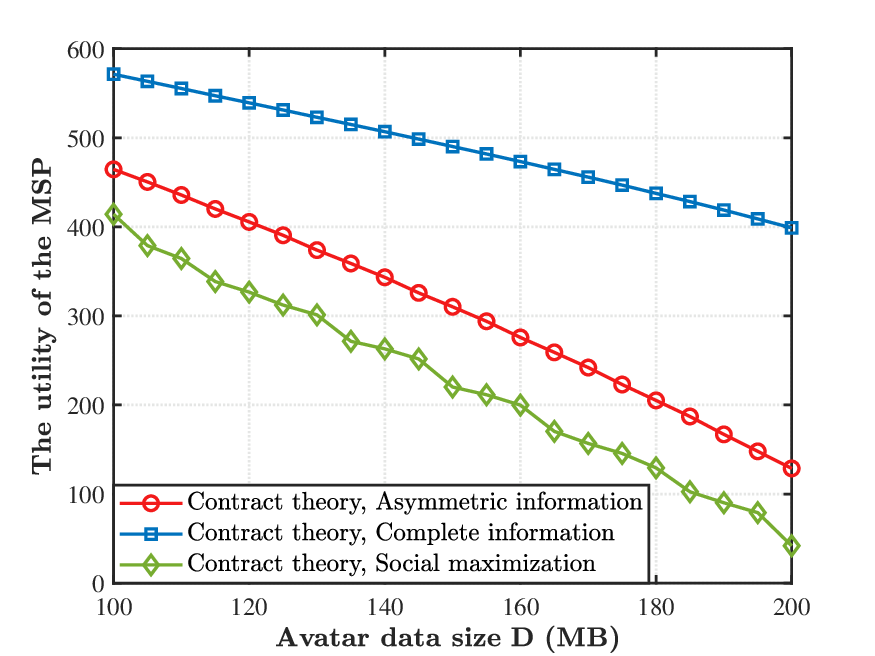}
		\caption{The utility of the MSP corresponding to different avatar data sizes $D$ under three incentive mechanisms.}\label{u_s}
	\end{minipage}
	\hspace{0.2in}
	\begin{minipage}[t]{0.48\linewidth}
		\centering
		\captionsetup{font={normal}}
		\includegraphics[width=1\linewidth]{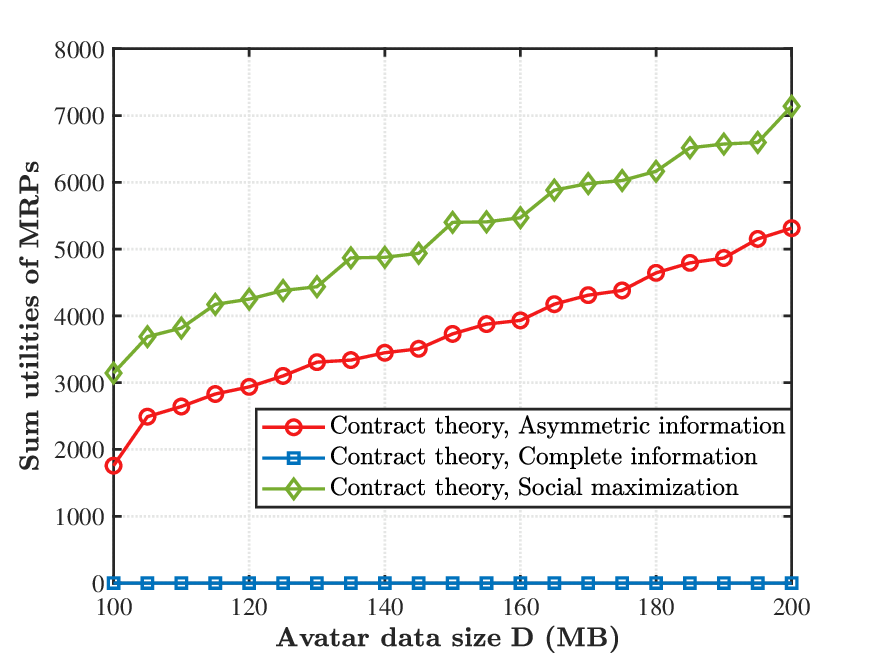}
		\caption{Sum utilities of MRPs corresponding to different avatar data sizes $D$ under three incentive mechanisms.}\label{u_w}
	\end{minipage}
    \end{center}
\end{figure*}

Figure \ref{IC_IR} shows the feasibility (i.e., IR and IC constraints) of the proposed scheme under information asymmetry. The utilities of four types of MRPs are shown when they sign different contract items. We can find that the utilities of MRPs are increasing with the increasing types of MRPs, and the utility of the MRP choosing the corresponding contract item is no less than $0$, which demonstrates that our designed contract guarantees the IR conditions. Besides, each MRP selects the contract item corresponding to its own type that achieves the maximum utility. For example, a type-$1$ MRP obtains the maximum utility only when it chooses the contract item $(b_1,R_1)$, which is exactly designed for its type. If the type-$1$ MRP selects any other contract items $(b_n,R_n),n\in\small\{2,\ldots,N\small\}$, its utility will reduce. Note that a similar phenomenon can be observed for all other types of MRPs when they choose the contract item designed for their corresponding types. Therefore, the above observations validate that our designed contract satisfies the IR and IC conditions. Based on the above analysis, we conclude that MRPs will automatically reveal their types to the MSP after choosing the contract item, which means that by utilizing the proposed scheme, the MSP can capture the MRPs’ private information and thus effectively alleviate the impact of information asymmetry.

Figure \ref{u_s} shows the utility of the MSP corresponding to different avatar data sizes $D$ under three incentive mechanisms. From Fig. \ref{u_s}, we can observe that regardless of the incentive mechanism, the utility of the MSP decreases as the avatar data size $D$ increases. The reason is that to meet the delay requirement of the avatar migration, the bigger avatar data size $D$ indicates that the MSP requires more bandwidth resources from MRPs and pays more rewards to them, thus decreasing the utility of the MSP. Besides, the utility of the MSP under the contract theory with complete information is always greater than that under the contract theory with asymmetric information, which indicates that the MSP obtains fewer benefits because of information asymmetry. The reason is that although the proposed scheme can effectively mitigate the effects of information asymmetry by leveraging contract theory\cite{hou2017incentive}, a rational MSP still has a chance to provide false information maliciously and cheat for more rewards, which decreases the utility of the MSP.

Figure \ref{u_w} shows the sum utilities of MRPs corresponding to different avatar data sizes $D$ under three incentive mechanisms. From Fig. \ref{u_w}, we can observe that as the avatar data size $D$ increases, the sum utilities of MRPs under the contract theory with complete information are always $0$, which indicates that the MRP receives rewards equal to its bandwidth cost with complete information. We can also find that the sum utilities of MRPs increase as the avatar data size $D$ increases under the contract theory with asymmetric information or the contract theory with social maximization. The reason is that since the amount of avatar data migrated increases, the MRPs can obtain more rewards based on the designed contract when they contribute more bandwidth resources for avatar migrations. Therefore, the sum utilities of MRPs increase as the avatar data size $D$ increases. Besides, the MRPs obtain the optimal utilities under the contract theory with social maximization, and the sum utilities of the MRPs under the contract theory with asymmetric information are greater than those under the contract theory with complete information.

\section{Conclusion}
In this paper, we have studied VT migrations in vehicular metaverses and formulated the incentive mechanism under asymmetric information for avatar migrations (as an example of VT migrations). We have proposed a novel metric named AoMT based on the concept of AoI for vehicular metaverses to quantify the task freshness of the avatar migration, which can evaluate the MSP's satisfaction. Furthermore, to improve the efficiency of avatar migrations, we have designed an AoMT-based contract model under information asymmetry for incentivizing MRPs to contribute bandwidth resources. Finally, numerical results have demonstrated the efficiency of the proposed incentive mechanism for avatar migrations in vehicular metaverses. In the future, we will improve the mathematical model to adapt to the VT migration. Besides, we may design a prototype system to evaluate our scheme and use artificial intelligence tools like deep reinforcement learning to enhance the solution methodology.

\bibliographystyle{IEEEtran}
\bibliography{ref}
\end{document}